\documentclass[11pt,oneside,a4paper]{article}

\usepackage{mathtools}
\usepackage{listings}
\usepackage{amsmath}
\usepackage{amsthm}
\usepackage{amssymb}
\usepackage{hyperref}
\hypersetup{colorlinks=true,
  citecolor=blue}
\usepackage{color}
\usepackage{float}
\usepackage{mathrsfs}
\usepackage{url}
\usepackage{amsfonts}
\usepackage[absolute]{textpos}

\newtheorem{prop}{{\bf Proposition}}[section]
\newtheorem{lem}{\bf Lemma}[section]
\newtheorem{thm}{\bf Theorem}[section]
\newtheorem{cor}{\bf Corollary}[section]

\newcommand{\ssup}[1]{^{^{_{#1}}}} %% small superscripts
\newcommand{\lsup}[1]{^{_{^{#1}}}} %% lower small superscripts
\newcommand{\ssub}[1]{_{_{^{#1}}}} %% small subscripts
\newcommand{\ds}{\kern .1em ds} %%%%%%%% measure
\newcommand{\dv}[1]{\kern .1em d#1} %%%% variable measure
\newcommand{\ddt}{\tfrac{d}{dt}} %%%% Small t- derivative text mode
\newcommand{\Exp}[1]{\exp\{#1\}}
\newcommand{\norm}[1]{\|#1\|} %% norm

\renewcommand{\Re}{\text{Re}\,}

\begin{document}
\numberwithin{equation}{section}

% \begin{textblock}{3}(3,0)
%   {\footnotesize {\color{red}Draft copy \#3.
%       For discussion purposes.}}
% \end{textblock}

\title{Epidemics: towards understanding\\ undulation and decay}

\author{Niko Sauer\\{\footnotesize University of Pretoria, South Africa}}
% \\{\footnotesize Centre for the Advancement of Scholarship}
%  \\{\footnotesize University of Pretoria}}

\date{}

\maketitle

\begin{abstract}\noindent\relax
  \small{Undulation of infection levels, usually called {\it waves},
    are not well understood.  In this paper we propose a mathematical
    model that exhibits undulation and decay towards a stable state.
    The model is a re-interpretation of the original SIR-model
    obtained by postulating different constitutive relations whereby
    classical logistic growth with recovery is obtained. The recovery
    relation is based on the premise that infectiousness only lasts
    for some time. This leads to a differential-difference (delay)
    equation which intrinsically exhibits periodicity in its solutions
    but not necessarily decay to asymptotically stable
    equilibrium. Limit cycles can indeed occur. An appropriate
    linearization of the governing equation provides a firm basis for
    heuristic reasoning as well as confidence in numerical
    calculations.}
\end{abstract}

\begin{quotation}
  \noindent
  \small{\textbf{MSC:} 92D30 --- Epidemiology; 34K13 --- Periodic
    solutions.}
\end{quotation}

%\vspace{5pt}

\begin{quotation}
  \noindent
  \small{{\color{blue}For my soul-mate Adri Prinsloo (1974--2021)
      whose penetrating questions ---why, not how--- contributed
      immensely to the development and understanding of this work. She
      left this life too soon.}}
\end{quotation}

\section{Introduction}
\label{Intro}

In a pioneering paper Hutchinson \cite{Hutchinson:1948} states: ``\dots
that circular paths often exist which tend to be self-correcting
within certain limits, but which break down, producing violent
oscillations \dots'' An equation to model this situation, is
given in a footnote to the paper as
\begin{equation}\label{hutch_eqn}
  \frac{dy}{dt}(t) = y(t)[1-y(t-\tau)]
\end{equation}
with $\tau$ a `time lag'. This, with some constants added, came to be
known as {\it Hutchinson's equation}. The oscillatory nature of its
solutions was the subject of a number of mathematical studies, the
earliest of which are Cunningham \cite{Cunningham:1954}, Wright
\cite{Wright:1955} and Jones \cite{Jones:1961}, sometimes in an
equivalent form. A generalization dealing with several time lags is
treated in Gopalsamy \cite{Gopalsamy:1992}. This was anteceded by Van
der Plank \cite{vdplank:1963}, \cite{vdplank:1965} who considered
plant diseases in which dormant as well as an infectious periods are
taken into account. The equation that carries his name, suitably
transformed, is
\begin{equation}\label{vdPlank}
  \frac{dy}{dt}(t) = Ry(t)[y(t-\tau\ssub{d}) - y(t-\tau\ssub{i})],
\end{equation}
which is of interest to us when the period $\tau\ssub{d}$ of dormancy is
taken to be zero. By and large, studies of equations such as these
have concentrated on long term behaviour of solutions, particularly
decay to a point of equilibrium; ``flattening of the curve''. A
phenomenological model for the prediction of ``waves'' is presented by
Cacciapaglia, Cot \& Sannino in \cite{Cacciapaglia:2022}. 

In present-day clinical contexts, the phenomenon of ``waves of
infection'' seems to be very much at the forefront but does not appear
to be well-understood or even defined. It is the purpose of this paper
to align this with Hutchinson's ``violent oscillations'' which have
been mathematically shown to be exhibited by equations such as
\eqref{hutch_eqn} and its generalizations. The term {\it wave} may be
inappropriate since wave phenomena inseparably involve both time and
space. For this reason we have chosen the word {\it undulation}. This
phenomenon is known to occur in plant as well as animal populations.

Instead of \eqref{hutch_eqn} we shall use the ``logistic delay
equation'' as stated in Ruan \cite{Ruan:2006}. It is of the form
$x'(t) = rx(t)[1 - a\ssub{1}x(t) - a\ssub{2}x(t-\tau)]$, with prime
denoting the time derivative $d/dt$. This equation turns up
occasionally without indication of the assumptions made to derive
it. In \S\ref{verhulst} we give a systematic derivation based on a
general view of the SIR model introduced by Kermack \& McKendick
\cite{Kermack:1927}. This view also leads to the ``theta-model'' which
can be used to obtain better correspondence to observed
data. Moreover, this approach establishes parameters with specific
significance which is easily lost when the treatment is entirely
mathematical. Section \ref{Norm} deals with normalization and scaling of
the equations to obtain more familiar forms. It is also shown there
that the theta-model equation can be transformed to the `standard'
form with constants and variables having different meanings.

Some general results are obtained in \S\ref{general} such as
positivity of solutions when the `initial history' is so, and an upper
bound which implies that solutions cannot grow in an unbridled
manner. In \S\ref{linearize} we obtain a `natural' linearization and
in \S\ref{Linear} treat the linear homogeneous problem with the aid of
the Laplace transform. This elaborate treatment serves to augment
rather sketchy treatments found in the literature, constantly keeping
track of the model parameters. Inversion of the Laplace transform is
discussed in \S\ref{Inversion}. This leads to a series representation
of the solution of the linear homogeneous problem. Section
\ref{estimates} gives sharp estimates of the position of poles and
terms in the series solution. It is shown that convergence to
equilibrium is guaranteed if a constant, given in terms of the
parameters, is sufficiently small.

The significance of the preceding analysis for the nonlinear problem
is discussed in \S\ref{nonlinear}. In \S\ref{example} we present a
numerical example to illustrate that undulations, as exhibited by the
model, can qualitatively be in accordance with observed phenomena. The
example shows that decay to an equilibrium point is possible, but also
that limit cycles are possible steady states. The numerical procedure,
based on a construction in \S\ref{general}, is given in
\S\ref{Numeric}.

The concluding remarks of \S\ref{Unscientific}, although unscientific,
have some seriousness about them.

\section{SIR, Verhulst and more}
\label{verhulst}

Fundamental to many mathematical descriptions of epidemics is the SIR
model and its variants. Here three quantities, the number of
susceptibles $S(t)$, the number of infectious individuals $I(t)$, and
the number recovered (restored?) $R(t)$ at time $t$, are related by
the `conservation principle'
\begin{equation}\label{conserv}
  S(t) + I(t) + R(t) = N
\end{equation}
with the constant $N$ denoting the total `population' considered. The
dynamics of the epidemic is in the system of ordinary differential
equations
\begin{align}
  &I'(t) = -[S'(t) + R'(t)];\label{I-prime}\\
  &S'(t) = \mathfrak{F}\ssub{S}(S(t),I(t),R(t));\label{S-prime}\\
  &R'(t) = \mathfrak{F}\ssub{R}(S(t),I(t),R(t)).\label{R-prime}
\end{align}
Equation \eqref{I-prime} is simply \eqref{conserv} differentiated.
For specific purposes the `driving forces' $\mathfrak{F}\ssub{S}$ and
$\mathfrak{F}\ssub{R}$ are chosen by postulates known as {\it
  constitutive relations}. In the original SIR-model, the postulates
are: $\mathfrak{F}\ssub{S} = - \beta I(t)[S(t)/N]$ and
$\mathfrak{F}\ssub{R} = \gamma I(t)$ with $\beta$ and $\gamma$
positive `rate' constants (dimension [time]$\ssup{-1}$).

The Verhulst logistic growth model \cite{Verhulst:1838}, originally
aimed at population growth, is sometimes used for epidemics when
recoveries do not occur. It may be considered as a fundamental
principle of population dynamics.  In this model the population size
$N$ is called the {\it carrying capacity}, the number of individuals
that can be infected, and recovery is ignored. Let $P(t)$ be the
probability of finding an infectious individual at time $t$. Then the
relevant constitutive relation, to which is added an assignment of 
probability, is expressed as follows:
\begin{align}
  &\mathfrak{F}\ssub{S} = -\beta I(t)[1-P(t)];\label{Verh_1}\\
  &P(t) = I(t)/N,\label{Verh_2}
\end{align}
where $1-P(t)$ signifies the probability of finding a `healthy'
individual.  This, combined with \eqref{I-prime}, yields the classical
{\it logistic equation}
\begin{equation*}
  I'(t) = \beta I(t)[1 - I(t)/N]
\end{equation*}
which is at the core of many informed speculations.

We now turn to the situation where {\it recovery} is also taken
into account and postulate the relation
\begin{equation}\label{Post}
  \mathfrak{F}\ssub{R} = \gamma I(t)P(t-\tau) ,
\end{equation}
with $\tau > 0$ the {\it period of infectiousness}. This means that
the rate of recovery is proportional to the probability of
infectiousness occurring at the earlier instant $t-\tau$. If $P(t)$ is
still specified according to \eqref{Verh_2}, the relations
\eqref{Verh_1}, \eqref{Post} gives the {\it logistic-recovery
  equation}
\begin{equation}\label{LR}
  I'(t) = I(t)\{\beta[1-(I(t)/N)]
    -\gamma(I(t-\tau)/N)\}.
\end{equation}  
This is the equation we shall study, although expressed differently. 

In a model for competing species Gilpin and Ayala \cite{Gilpin:1973}
essentially chose $P(t) = [I(t)/N]^\theta$ with $\theta$ a
positive constant. In the context of our discussion here, this gives
rise to the {\it theta-logistic-recovery equation}
\begin{equation}\label{TLR}
   I'(t) = I(t)\{\beta[1-(I(t)/N)^{\theta}]
     -\gamma(I(t-\tau)/N)^{\theta}\}.
\end{equation} 
%%
% \begin{equation*}\label{theta}
%   I'(t) = \beta I(t)\{1 - [I(t)/N]^{\theta}\}.
% \end{equation*}
%%
One may view this as a mathematical generalization of the logistic
model ($\theta = 1$) to manipulate the sigmoidal curve, but it can be
grounded in probability theory (Feller \cite[II.5--8,
Randomization]{Feller:1966}).

Constitutive relations in SIR models can be found in Della Morte \&
Sannino \cite{DellaMorte:2021} and Buonomo \& Cerasuolo
\cite{Buonomo:2015}. In the latter a time lag is introduced in 
$\mathfrak{F}\ssub{S}$. A SIR-model with delay in
$\mathfrak{F}\ssub{R}$, not the same as \eqref{Post}, has been
suggested by Reiser \cite{Reiser:2020}.

\section{Normalization,  scaling and a reduction}
\label{Norm}

We first normalize the equation \eqref{LR} by the setting 
$F(t) = I(t)/N$ so that $I = N$ corresponds to $F = 1$, and
$P(t)=F(t)$. One may think of $F(t)$ as the {\it level of infectiousness}
at time $t$. The result is
\begin{equation}\label{Normal}
  F'(t) = F(t)\{\beta[1-F(t)] - \gamma F(t-\tau)\}.
\end{equation}
This is a {\it differential-difference equation} or {\it delay
  equation}. To solve the equation for times $t>0$ we need to know the
state (history) of $F$ for times $t \in [-\tau,0]$ (see {\it e.g.},
Bellman \& Cooke \cite{Bellman:1963}). Thus we have the {\it
  initial condition}
\begin{equation}\label{Init-cond}
  F(t) = \Phi(t)\ \text{for }-\tau \le t \le 0,
\end{equation}
with $\Phi$ a given function defined on $[-\tau,0]$.

If it is assumed that an asymptotically, nonzero stable state for
solutions of \eqref{Normal} exists, that is, if $F'(t) \to 0$ as
$t \to \infty$ and the limit $F\ssub{\infty} = \lim_{t\to\infty} F(t)$
exists, it follows that $F\ssub{\infty} = \beta/(\beta + \gamma)$.  We
use this parameter to scale the equation \eqref{Normal}. In addition
the parameter $\tau$ is used as unit of time to obtain a completely
dimensionless formulation. Thus we define the new variables
$t\ssub{\tau} = t/\tau$ and
$f(t\ssub{\tau}) = F(t)/F\ssub{\infty} = F(\tau
t\ssub{\tau})/F\ssub{\infty}$ to obtain, in place of \eqref{Normal},
% %%
\begin{align}
  f'(t\ssub{\tau}) &= \beta\tau f(t\ssub{\tau}) \left\{\left[1 -
          \left(\frac{\beta}{\beta+\gamma}\right)f(t\ssub{\tau})\right]
          - \left(\frac{\gamma}{\beta+\gamma}\right)f(t\ssub{\tau}-1)\right\}
          \notag\\
        &= \tilde{\beta} f(t\ssub{\tau}) \{[1 - \beta\ssup*f(t\ssub{\tau})]
          - \gamma\lsup* f(t\ssub{\tau}-1)\}\label{expr1}
\end{align}
with $\tilde\beta := \beta\tau$,
$\beta\ssup* := \beta/(\beta+\gamma)$,
$\gamma\lsup* := \gamma/(\beta+\gamma)$ and the prime denoting
$d/dt\ssub{\tau}$.  We note also the identity
$\beta\ssup* + \gamma\lsup* = 1$ which will be important at a later
stage and gives rise to an alternative form discussed in
\S\ref{Unscientific}. The initial condition \eqref{Init-cond} now has
the form
\begin{equation}\label{init-cond}
  f(t\ssub{\tau}) = \phi(t\ssub{\tau})
      := \frac{\Phi(\tau t\ssub{\tau})}{F\ssub{\infty}}
       \ \text{for }-1 \le t\ssub{\tau} \le 0.
\end{equation}
It is clear that if a non-zero stable asymptote exists, it has the
value $f\ssub{\infty} = 1$. Also note that \eqref{expr1} has the same
form as the ``logistic delay equation'' mentioned in \S\ref{Intro} but
the constants can be interpreted in terms of the constitutive
relations \eqref{Verh_1}, \eqref{Post}.

Under the transformations just described, the equation \eqref{TLR} has
the form
\begin{equation*}
  f'(t\ssub{\tau}) = \tilde{\beta}
  f(t\ssub{\tau})\{1-\beta\ssup{*}f^{\theta}(t\ssub{\tau} )
     - \gamma\lsup{*}f^{\theta}(t\ssub{\tau}-1)\}.
\end{equation*}
This equation can assume a less intimidating look if we replace
$f^{\theta}(t\ssub{\tau})$ with $f(t\ssub{\tau})$ to obtain an
equation of precisely the same form as \eqref{expr1} except that now
$\tilde{\beta} = \theta\beta\tau$. Thus we go on to study
\eqref{expr1}.

\section{Some general considerations}
\label{general}

In this section we consider the scaled equation \eqref{expr1} together
with the initial condition \eqref{init-cond}. We shall abuse notation
by writing $t$ instead of the dimensionless time $t\ssub{\tau}$. The
initial level $f\ssub{0} := f(0)$ will be of significance. 

First we obtain a formal representation of the solution by the
substitution $g(t) = 1/f(t)$, familiar for equations of the
Bernoulli-kind. This leads to
\begin{equation}\label{Bernoulli}
  g'(t) + \tilde{\beta}[1 - \gamma\lsup{*}f(t-1)]g(t) =
     \beta\ssup{*}\tilde{\beta}.
\end{equation}
If we interpret the term $f(t-1)$ as $1/g(t-1)$ this is a
differential-difference equation for $g$ with initial condition
\begin{equation}\label{initg}
  g(t) = 1/\phi(t) \text{ for } t \in [-1,0].
\end{equation}
Associated with this equation we have the integrating factor
\begin{equation}\label{intf1}
  i(t) := \tilde{\beta}[t - \gamma\lsup{*}\psi(t)];\
    \text{ with } \psi(t) := \int_0^t f(s-1)\dv{s} 
      = \int_{-1}^{t-1} f(s)\dv{s}, 
\end{equation}
so that
\begin{equation}\label{tranfd_eqn}
  \ddt[\exp\{i(t)\}g(t)\}] = \beta\ssup{*}\tilde{\beta}\exp\{i(t)\}.
\end{equation}
This can be integrated directly. However, further integration by parts
of the term on the right yields:
\begin{align}
  \exp\{i(t)\}g(t) &= g\ssub{0} - \beta\ssup{*}
     + \beta\ssup{*} \exp\{i(t)\}\notag\\
     &\qquad + \tilde{\beta}\beta\ssup{*}\gamma\ssup{*}
         \int_0^t\exp\{i(s)\}f(s-1)\dv{s},\label{rep1}
\end{align}
where $g\ssub{0} = 1/f\ssub{0}$.

The formal calculations above can be placed on a firmer footing under
the following hypotheses about the initial state which will be taken
for granted from now on:
\begin{enumerate}
  \item[H1.] {\sl The initial state $\phi$ is continuous on} $[-1,0]$.
  \item[H2.] $\phi(t) > 0$ {\sl for} $t \in [-1,0]$.
\end{enumerate}  
We immediately note that the integrating factor $i(t)$ exists and is
differentiable for $t > 0$.

The initial value problem we study is well-posed in the sense of the
following result.

\begin{thm}\label{positivity}
  Under the assumptions \textnormal{H1} and \textnormal{H2}:

  \begin{enumerate}
    \item[\textnormal{(a)}] The function $g(t)$ as represented in
      \eqref{rep1} is a positive solution of \eqref{Bernoulli},
      \eqref{initg} with $g(0) = g\ssub{0}$.
    \item[\textnormal{(b)}] The function $f(t) = 1/g(t)$ is a positive
        solution of \eqref{expr1}, \eqref{init-cond}. 
  \end{enumerate}
\end{thm}
\begin{proof}
  The proof of the two assertions will simultaneously unfold by
  progression over the time intervals $[0,1), [1,2), \dots$ as in many
  instances to be found in \cite{Bellman:1963}.

  We begin with $0 \le t < 1$. Here the function $\psi$ is totally
  determined by the initial state $\phi$; it is in fact differentiable
  and positive. Thus the integrating factor is of suitable nature and
  the function $g(t)$, as determined by \eqref{rep1}, is positive and
  solves the initial value problem \eqref{Bernoulli},
  \eqref{initg}. It follows that $f(t) > 0$ solves \eqref{expr1},
  \eqref{init-cond}. Thus (a) and (b) are established in
  $[0,1)$. In addition the limit as $t \to 1$ defines the functions
  $g(1)$ and $f(1)$

  The same argument can be followed in the interval $[1,2)$,
  as the crucial properties have been established in $[0,1]$. It is
  clear that a formal induction argument will lead to the required
  outcome. 
\end{proof}
The next result shows that there are limits to infectivity levels and
that there is at most one stable asymptote.  
%%
%\pagebreak

\begin{thm}\label{Upper-bd}
  Under the assumptions \textnormal{H1} and \textnormal{H2} the
  solution of the initial value problem \eqref{expr1},
  \eqref{init-cond} is restricted in the following ways:

  \begin{enumerate}
    \item[\textnormal{(a)}] It is bounded. Specifically,
      \begin{equation*}
       0 < f(t) < \frac{f\ssub0}{\beta\ssup* f\ssub0
         + (1-\beta\ssup* f\ssub0)\exp\{-\tilde{\beta}t\}}
           \text{ for } t >0.
      \end{equation*}
    \item[\textnormal{(b)}] If the limit $f\ssub{\infty} =
      \lim_{t\to\infty}f(t)$ exists, it is equal to 1
  \end{enumerate}
\end{thm}
\begin{proof}
  Since $f(t) > 0$ for $t \ge -1$, we see from \eqref{intf1} that
  $i(t) < \tilde{\beta} t$. From the representation
    \eqref{rep1} follows that
    \begin{equation*}
      g(t) > \beta\ssup{*} + (g\ssub{0} - \beta\ssup{*})
             \exp\{-i(t)\}
           > \beta\ssup{*} + (g\ssub{0} - \beta\ssup{*})
             \exp\{-\tilde{\beta}t\}.
    \end{equation*}
  The stated upper bound is obtained by reciprocation and further
  manipulation.

  To prove (b) let us assume that the limit is zero. Then for given
  $\varepsilon \in (0,1)$ there exists $t\ssub{0}$ such that for
  $t > t\ssub{0}$ both $f(t)$ and $f(t-1)$ are less than
  $\varepsilon$. This implies that
  $f'(t) > \tilde{\beta}f(t)[1-\varepsilon]$ since
  $\beta\ssup*+\gamma\lsup* = 1$. It now follows that
  $f(t) >
  f(t\ssub{0})\exp\{\tilde{\beta}[1-\varepsilon](t-t\ssub{0})\}$ for
  all $t > t\ssub{0}$ which is absurd.

  It is seen from \eqref{expr1} that the limit $\lim_{t\to\infty}f'(t)$
  exists. The mean value theorem shows that this limit is zero, and
  hence $f\ssub{\infty}[1 - f\ssub{\infty}] = 0$. 
\end{proof}
We note that $f(t) < M := \max\{f\ssub{0},1/\beta\ssup*\}$. 

\section{Linearization}
\label{linearize}

For the problem at hand we arrive at a suitable linearization by
shifting the (expected) equilibrium level from $f = 1$ to $f = 0$. To
avoid an undue proliferation of symbols, we once again abuse notation by 
the replacing $f(t)-1$ by $f(t)$.  The governing equation
\eqref{expr1} then has the form
\begin{equation}\label{split}
  f'(t) + \tilde{\beta}[\beta\ssup{*}f(t) + \gamma\lsup{*}f(t-1)]
    = -\tilde{\beta}f(t)[\beta\ssup{*}f(t) + \gamma\lsup{*}f(t-1)]. 
\end{equation}
Of course, the initial condition \eqref{init-cond} is adapted
accordingly.

One immediately notes that the left of \eqref{split}, as opposed to
the right, is linear. Intuitively, if $f(t)$ is close to zero for
large $t$, the nonlinear term which is quadratic in $f$ will be
insignificant in the long run. Formally, the right of \eqref{split}
linearizes to zero. This leads us to the linear homogeneous equation
\begin{equation}\label{linzd}
  f'(t) + \tilde{\beta}[\beta\ssup{*}f(t) + \gamma\lsup{*}f(t-1)]
     = 0,
\end{equation}
which will be studied in detail.

Some simplifying notation is introduced: $b := \tilde{\beta}\beta\ssup{*}$,
$c := \tilde{\beta}\gamma\lsup{*}$ and $B := c\exp\{b\}$. 

We obtain an estimate for solutions of the non-homogeneous problem
\begin{equation}\label{non-hom}
  \left.
    \begin{aligned}
      &f'(t) + bf(t) + cf(t-1) = v(t) \text{ for } t > 0;\\
      &f(t) = \phi(t) \text{ for } -1 \le t \le 0,
    \end{aligned}
  \right\}
\end{equation}
with $v$ a given function, continuous on $[0,\infty)$.

\begin{thm}\label{est_non-hom}
   Every solution of \eqref{non-hom} satisfies
  \begin{equation*}
    |f(t)| \le [F\ssub{0} + V(t)]\exp\{(B - b)t\}; \quad t \ge 0,
  \end{equation*}
  with $F\ssub{0} := |f\ssub{0}| +
  c\int_{-1}^0\exp\{bs\}|\phi(s)|\dv{s}$ and  $V(t) = \int_0^t\exp\{bs\}
  |v(s)|\dv{s}$.
\end{thm}
\begin{proof}
  Let us write the first equation in \eqref{non-hom} in the form
  \begin{equation*}
    \ddt[\exp\{bt\}f(t)] = \exp\{bt\}[v(t) - cf(t-1)],
  \end{equation*}
  integration of which yields
  \begin{align*}
   &[\exp\{bt\}f(t)] = f\ssub{0} - c\int_0^t\exp\{bs\}f(s-1)\dv{s}
                       + \int_0^t\exp\{bs\}v(s)\dv{s}\\
   &\qquad = f\ssub{0} - c\int_{-1}^0\exp\{bs\}\phi(s)\dv{s} 
                       -c\exp\{b\}\int_0^{t-1}\exp\{bs\}f(s)\dv{s}\\
                     &\qquad\qquad + \int_0^t\exp\{bs\}v(s)\dv{s}.
  \end{align*}
  With $F(t) := \exp\{bt\}|f(t)|$, 
we obtain the integral inequality
  \begin{equation*}
    F(t) \le F\ssub{0} + V(t) + B\int_0^{t}F(s)\dv{s}.
  \end{equation*}
  Since $V$ is monotonically increasing, the Gr\"onwall-Bellman
  inequality \cite{Bellman:1943}, \cite[Lemma 3.1]{Bellman:1963}
  applies, {\it i.e.}, $F(t) \le [F\ssub{0} + V(t)]\exp\{Bt\}$,
  and the proof is complete.
\end{proof}

This result shows that the solution of \eqref{non-hom} is unique and
depends continuously on the initial data.

\section{The linear homogeneous equation}
\label{Linear}

Our attention now turns to the linear equation \eqref{linzd} under the
initial condition
$f(t) = \phi(t) \text{ for } -1 \le t \le 0; f(0) = f\ssub{0}$. The
approach is by the Laplace transform defined as
$\hat{f}(s) := \int_0^\infty \exp\{-st\}f(t)\dv{t}$ for complex $s$
with positive real part. In fact, Thm. \ref{est_non-hom} with
$v \equiv 0$ shows that for $\Re s > B-b$ the Laplace transform of
\eqref{linzd} may be taken. The result is
\begin{equation}\label{Laplaced}
  \left.
    \begin{aligned}
      &h(s)\hat{f}(s) = H(s);\\
      &h(s) := s + b + c\exp\{-s\};\\
      &H(s) := f\ssub{0} - c\exp\{-s\}
      \int_{-1}^0\exp\{-st\}\phi(t)\dv{t}.
    \end{aligned}
  \right\}
\end{equation}

For the inversion of $\hat{f}$ we need to study the zeros of the
complex-valued function $h(s)$ for $s = x + iy$ ($i = \sqrt{-1}$) in
the complex plane. For this the real and imaginary parts of $h$ must 
vanish and we have the equations
\begin{align}
  &(x+b) + c\exp\{-x\}\cos y = 0;\label{i}\\
  &y - c\exp\{-x\}\sin y = 0,\label{ii}
\end{align}
both of which need to be satisfied.

As a first step we eliminate the trigonometric terms in \eqref{i},
\eqref{ii} to obtain
\begin{equation}\label{Curve^2}
  y^2 = c^2\exp\{-2x\} - (x+b)^2.
\end{equation}
This defines a curve on which the solution points $s = x + iy$ must lie,
but not every point on the curve is necessarily a solution of
\eqref{i} and \eqref{ii}. Moreover, the curve so obtained is only
defined for those $x$ for which the right of \eqref{Curve^2} is
non-negative. We investigate this question first.

\begin{lem}\label{bounds}
  With $B$ a positive constant:
  \begin{enumerate}
  \item[\textnormal{(a)}] The equation $B\exp\{-u\} = u$ has a unique
    positive solution $u\ssub{M}$ and $B/(1+B) \le u\ssub{M} < B$.
  \item[\textnormal{(b)}] If $B > e^{-1}$ the equation
    $B\exp\{u\} = u$ has no positive solution.
  \end{enumerate}
\end{lem}
\begin{proof}
  Assertion (a): The existence, uniqueness, positivity and the upper
  bound of a solution is straightforward. To obtain the lower bound we
  note that $\exp\{-u\} \ge 1-u$ for $u>0$. Hence
  $u\ssub{M} = B\exp\{-u\ssub{M}\} \ge B(1-u\ssub{M})$ and the result
  follows.

  Assertion (b) follows from the inequality $u\exp\{-u\} \le e^{-1}$.
\end{proof}

In accordance with \eqref{Curve^2}, let
$
  k\ssup{(2)}\kern-3pt(x) := c^2\exp\{-2x\} - (x+b)^2
                  = \break[c\exp\{b\}]^2\exp\{-2(x+b)\} - (x+b)^2.
$
\begin{prop}\label{k-positivity}
  Let $B = c\exp\{b\}$ and $u\ssub{M}$ as in Lemma \ref{bounds}. Then

  \begin{enumerate}
   \item[\textnormal{(a)}] For $-b \le x \le x\ssub{M} := u\ssub{M} - b$,
     $k\ssup{(2)}\kern-3pt(x) \ge 0$ with equality only if $x =
     x\ssub{M}$.
   \item[\textnormal{(b)}] For $x+b < 0$ the function
     $k\ssup{(2)}\kern-3pt(x) > 0$ if $B > e^{-1}$.
  \end{enumerate}
\end{prop}
\begin{proof}
  For statement (a), let $u = x+b \ge 0$. Then
  $k\ssup{(2)}\kern-3pt(x) = B^2\exp\{-2u\} - u^2 =
    [B\exp\{-u\}+u][B\exp\{-u\}-u]$ which, by Lemma \ref{bounds}(a) is
    non-negative for $u \le u\ssub{M}$. To prove (b) let $u =
    -(x+b) > 0$. Then $k\ssup{(2)}\kern-3pt(x) =  B^2\exp\{2u\} - u^2$. The
    result follows from Lemma \ref{bounds}(b).
\end{proof}

It is of importance to introduce the parameter
$\rho := \gamma/\beta = \gamma\lsup{*}/\beta\ssup{*}$. It corresponds
to the reciprocal of the {\it basic reproduction number}
$R\ssub{0} = \beta/\gamma$ in the classical SIR model. This leads to
the relation $c = \rho b$. From now on we make an assumption stronger
than suggested by Lemma \ref{bounds}(b) namely
\begin{equation}\label{stronger}
B = c\exp\{b\} = \rho b \exp\{b\} > 1.
\end{equation}
Under this assumption the function $k\ssup{(2)}\kern-3pt(x) \ge 0$ for
$x \le x\ssub{M}$ so that $k\ssup{(2)}\kern-3pt(x)$ can
actually be considered  a square. We define the function $k$ by
\begin{align}
  k^2(x) := k\ssup{(2)}\kern-3pt(x) &= c^2\exp\{-2x\} -(x+b)^2\notag\\
    &\qquad= B^2\exp\{-2(x+b)\} -(x+b)^2
      \notag\\
    &\qquad= \rho^2 b^2 \exp\{-2x\} - (x+b)^2;\quad x \le x\ssub{M}.
      \label{def_k}
\end{align}
The curve $K$ defined in the complex plane by $s = x \pm ik(x)$ will
be our next concern. Because of symmetry we shall deal mainly with the
positive branch which will also be referred to as $K$.

\begin{prop}\label{Properties}
  The curve $K$ has the following properties:
  \begin{enumerate}
    \item[\textnormal{(a)}] $k(x) \to \infty$ as $x \to -\infty$.
    \item[\textnormal{(b)}] At $x < x\ssub{M}$ the tangent is negative. 
    \item[\textnormal(c)] $x\ssub{M}$ is positive if and only if
      $\gamma > \beta$, \text{i.e.,} $\rho > 1$.
  \end{enumerate}  
\end{prop}
\begin{proof}
  The assertion (a) is, by \eqref{def_k}, straightforward.
  
  From \eqref{def_k} we see that $k(x)k'(x) =
  -[B^2\exp\{-2(x+b)\} + (x+b)]$. If $x+b > 0$, the term in brackets
  is positive. Since $k(x) > 0$ it follows that $k'(x) < 0$ for such
  $x$. If $x+b < 0$ let $u = -(x+b)$ and it follows that
  $B^2\exp\{-2(x+b)\} + (x+b) = B^2\exp\{2u\} - u > B^2(1+2u)-u >
  1+u > 0$, by the assumption \eqref{stronger}. Thus (b) is
  established.

  To prove (c) we notice that $k\ssup{(2)}\kern-3pt(0) = c^2 -
  b^2$. Hence if $c > b$, the point $(0,[c^2-b^2]^{1/2})$ is on the
  curve $K$. From (b) we see that $x\ssub{M}$, where $k(x) = 0$, must
  be positive. But $c > b$ means the same as $\gamma > \beta$. This
  argument can also be reversed.
\end{proof}

From what we have established so far, the following is significant:
\begin{enumerate}
\item[A.] Under the assumption \eqref{stronger} the zeros of the
  function $h$ are all to the left of the vertical line $x = B$.

\item[B.] If $\beta \ge \gamma$ the zeros in question are to the left of
  the line $x = \delta$ for arbitrary $\delta > 0$. In fact, if $\beta
  = \gamma$, $x\ssub{M} = 0$.
\end{enumerate}
Thus inversion of the Laplace transform becomes a distinct possibility
if the zeros on the curve $K$ can be located. For that we need to
obtain information about the points $(x,k(x))$ for which the equations
\eqref{i} and \eqref{ii} are actually satisfied. Towards this we
eliminate the exponential terms from these equations to obtain the
relation $y + (x+b)\tan y = 0$.
%\end{equation*}
%%
additional to \eqref{Curve^2}. This defines another (multi-branched)
curve $L$ which has to meet the curve $K$ in certain points. The
function to be considered is
\begin{equation}\label{Add_curve}
  \ell(x) := -(x+b)\tan k(x).
\end{equation}
The zeros we look for will occur at points where $\ell(x) =
k(x)$.  It is seen from Prop. \ref{Properties}(a) that this will
happen at points where $k(x)$ is near $(n+1/2)\pi; n =
0,1,2,\dots$ and this results in a discrete sequence of zeros of
$h$, each of the form $s = x + ik(x)$.

We next examine the case $x+b \ge 0$ where zeros with non-negative
real part may occur. According to Prop. \ref{Properties}(c) this can
only happen if $\rho > 1$. It will be necessary to indicate the
dependence of functions and derived parameters on
$b = \tilde{\beta}\beta\ssup{*} = \beta\ssup{*}\beta\tau$ and
$\rho = \gamma/\beta$. Thus we write $B = B(b,\rho)$,
$k(x) = k(x;b,\rho)$ in accordance with \eqref{stronger},
\eqref{def_k}. Also note that $x\ssub{M}$ also depends on $b$ and
$\rho$. 

It is convenient to consider equation \eqref{i} on the curve $K$
instead of \eqref{Add_curve}. This yields, after some manipulation,
\begin{equation}\label{Alt_eqn}
  r(x;b,\rho) := (x+b)\exp\{x+b\} + B(b,\rho)\cos k(x;b,\rho) = 0.
\end{equation}
To investigate this equation we consider
$r(-b;b,\rho) = B(b,\rho)\cos k(-b;b,\rho) = B(b,\rho)\cos B(b,\rho)$
as can be seen from \eqref{def_k}. Thus, if
$3\pi/2 > B(b,\rho) > \pi/2 > 1$, $r(-b;b,\rho) < 0$.  Also at
$x = x\ssub{M}$, we find that
$r(x\ssub{M};b,\rho) = (x\ssub{M}+b)\exp\{x\ssub{M}+b\} + B(b,\rho) >
0$ since $k(x\ssub{M};b,\rho) = 0$. We conclude that $r(x;b,\rho)$ has
zeros in the interval $(-b,x\ssub{M})$. These zeros may still be
negative.

We can, however,
find a very interesting value of $b$ by noticing that $k(0;b,\rho) =
(\rho^2-1)^{-1/2}b$ and  
\begin{equation}\label{at_zero}
  r(0;b,\rho) = b\exp\{b\}[1 + \rho\cos([\rho^2-1]^{1/2}b)].
\end{equation}
Thus, if we take $b$ as
\begin{equation*}%\label{crit-b}
  b\ssub{0} = b\ssub{0}(\rho) := (\rho^2 - 1)^{-1/2} \arccos(-\rho^{-1})
   = (\rho^2 - 1)^{-1/2}[\pi - \arccos(\rho^{-1})],
\end{equation*}
it is seen that $r(0;b\ssub{0},\rho ) = 0$. Moreover, since
$\rho > 1$, $0 < \arccos(\rho^{-1}) < \pi/2$ so that
$\pi/2 < (\rho^2-1)^{1/2}b\ssub{0} < \pi$. Thus,
$\rho b\ssub{0} > \pi/2$ and hence
$B(b\ssub{0},\rho) > (\pi/2)\exp\{b\ssub{0}\} > \pi/2$. For this
particular choice of $b$, $x=0$ therefore is a zero of
$r$. Corresponding to $b\ssub{0}$ is a critical value of $\tau$:
\begin{equation}\label{crit-tau}
  \tau\ssub{0} = {b\ssub{0}}/{\beta\beta\ssup{*}}
    = (1+\rho)b\ssub{0}(\rho)/\beta.
\end{equation}

A positive zero may be contrived by incrementing $b\ssub{0}$ without
changing $\beta$ and $\gamma$ ({\it i.e.}, $\rho$ fixed). This amounts
to letting $\tau = \tau\ssub{0} +\sigma$; $\sigma >0$. Then 
$b = b\ssub{0} + \beta\ssup{*}\beta\sigma$. From
\eqref{at_zero} we see that if
\begin{equation*}
  \beta\sigma \le \left[\frac{\rho+1}{\rho-1}\right]^{1/2}
    [\pi-(\rho^2-1)^{1/2}b\ssub{0}],
\end{equation*}
$r(0;b) < 0$ which means that a positive zero of \eqref{Alt_eqn} 
exists.

The zeros of $h$ are simple. Indeed, if $s\lsup{*}$ is a zero, then
$s\lsup{*} + b = -c\exp\{-s\lsup{*}\}$ and
$h'(s\lsup{*}) = 1 - c\exp\{-s\lsup{*}\} = 1 + b + s\lsup{*}$ which
cannot be zero. Also, when the negative branch of $K$ namely,
$s = x - ik(x)$ is considered, we see from \eqref{Add_curve} that if
$s = x +ik(x)$ is a zero on the positive branch, its complex conjugate
$\overline{s} = x - ik(x)$ is a zero on the negative branch.

Positioning of the zeros of $h$ is illustrated in
Fig.\ref{Zeros}. Increasing $\tau$ could shift $x\ssub{0}$ to `the
other side'. We order the roots of $h$ according to their real parts:
$x\ssub{n} > x\ssub{n+1}; n=0,1,\dots$ and note that
$x\ssub{n} \to -\infty$ as $n \to \infty$.
\begin{figure}[H]
 \centering
 \includegraphics[width=0.6\textwidth]{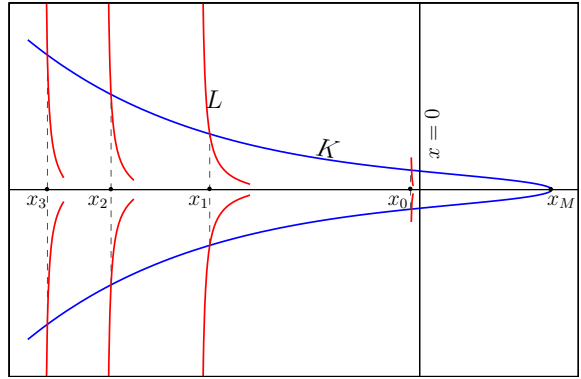} 
 \caption{{\footnotesize Zeros of $h(s)$}}
 \label{Zeros}
\end{figure}

We summarize the findings above:
\begin{thm}\label{About_roots}
  In terms of the parameters $b = \tilde{\beta}\beta\ssup{*}$,
  $\rho = \gamma/\beta$ and $B(b) = \rho b\exp\{b\}$ the following is
  known about the zeros of the function
  $h(s) = s + b[1+\rho\exp\{-s\}]$:
  
  \begin{enumerate}
  \item[\textnormal{(a)}] All zeros are simple.
  \item[\textnormal{(b)}] There is a constant $p > 0$ such that all
    zeros occur to the left of the contour
    $C = \{s = p + iy: y \in \mathbb{R}\}$.
  \item[\textnormal{(c)}] If $B(b,\rho) > 1$ there is a decreasing
    unbounded sequence
    $\{x\ssub{n} < x\ssub{M}: n = 0, 1, 2,\dots\}$ such that
    $s\ssub{n} = x\ssub{n} + ik(x\ssub{n};b)$ are zeros. The complex
    conjugates are also roots.
  \item[\textnormal{(d)}]If $\rho > 1$ and $3\pi/2 > B(b,\rho) > \pi/2$,
    a finite number zeros with non-negative real part can occur. If
    $\rho \le 1$ there are only zeros with negative real part. 
  \end{enumerate}
\end{thm}

\section{Laplace inversion}
\label{Inversion}

The significance of Thm. \ref{About_roots} is seen by considering the
Mellin inversion of $\hat{f}$. From
\eqref{Laplaced}, $\hat{f} = H(s)/h(s)$ and the inversion yields the
representation
\begin{equation*}
  f(t) = \frac{1}{2\pi i} \int_C \frac{H(s)}{h(s)} \exp\{st\}\dv{s}.
\end{equation*}
From the residue theorem the integral on the right equals $2\pi i$
times the sum of the residues of the integrand at its poles, and the
poles are precisely the zeros of $h$. Now, let
$s\ssub{n} = x\ssub{n} + iy\ssub{n} = x\ssub{n} + ik(x\ssub{n})$ be one
of the zeros of $h$ obtained from the positive branch of the curve $K$
as discussed so far. To calculate the residue, we consider the Taylor
expansion of $h$ about $s\ssub{n}$. By taking into account that
$s\ssub{n} + b = -c\exp\{-s\ssub{n}\}$, $h'(s) = 1 + b -c\exp\{-s\}$
and $h\ssup{(j)}(s) = (-1)^jc\exp\{-s\}$ for $j \ge 2$, we obtain
\begin{equation*}
  h(s) = (s-s\ssub{n})[1+b+s\ssub{n}
    + (b+s\ssub{n}) \sum_{j=2}^\infty (-1)^{j-1}
      \tfrac{(s-s\ssub{n})^{j-1}}{j!}].
\end{equation*}
The pole at $s=s\ssub{n}$ therefore has the residue
\begin{equation*}
  R(s\ssub{n}) = \tfrac{H(s\ssub{n})}{1+b+s\ssub{n}}
    \exp\{s\ssub{n} t\}.
\end{equation*}
However, the negative branch of the curve $K$ also contributes. In
fact, if $s\ssub{n} = x\ssub{n} + iy\ssub{n}$ with
$y\ssub{n} = k(x\ssub{n})$, its complex conjugate $\overline{s}\ssub{n}$ is
also a zero of $h$. The residue at this pole turns out to be the
complex conjugate of $R(s\ssub{n})$. If we write
$H(s\ssub{n})/(1+b+s\ssub{n}) = |H(s\ssub{n})|/|1+b+s\ssub{n}|
\exp\{i\sigma\ssub{n}\}$, the two residues together contribute to the
solution by the term
\begin{equation}\label{term}
  f\ssub{n}(t) = 2\frac{|H(s\ssub{n})|}{|1+b+s\ssub{n}|}
  \exp\{x\ssub{n}t\}\cos(y\ssub{n}t+\sigma\ssub{n}).
\end{equation}
We immediately note that, since $y\ssub{n} > 0$, there is undulation
in every such term. Also, if $x\ssub{n} < 0$ the term decays to zero
exponentially. This is not the case when $x\ssub{n} \ge 0$. Thus the
dominant term in the solution will correspond to
$s\ssub{0} = x\ssub{0} + ik(x\ssub{0})$.

We conclude this section by estimating the period $T$ of the principal
mode of undulation, namely that associated with $x\ssub{0}$. The
`angular velocity' is $y\ssub{0} = k(x\ssub{0})$ and the scaled period
$T\ssub{\tau} = 2\pi/y\ssub{0}$ so that in unscaled time,
\begin{equation}\label{period}
  T = T\ssub{\tau}\tau = \left(\frac{2\pi}{y\ssub{0}}\right)\tau.
\end{equation}
The value of $x\ssub{0}$ can be obtained numerically by (carefully)
solving \eqref{Alt_eqn} for $x$, making sure that the obtained value
is the largest.

\section{Estimates}
\label{estimates}

The aim of this section is to obtain information about long-term
behaviour of the solution of the homogeneous equation
$f'(t) + bf(t) +cf(t-1) = 0$ under the initial condition
$f(t) = \phi(t)$ for $t \in [-1,0]; f(0) = f\ssub{0}$. To begin with
we notice that, at least formally,
$f(t) = \sum_{n=0}^\infty f\ssub{n}(t)$ with the terms $f\ssub{n}$
given by \eqref{term}. Since $x\ssub{n} \to -\infty$, there is a
smallest $m \ge 0$ such that $x\ssub{m} < -b$. With is in mind, we
define the (possibly) {\it principal part} of the solution as
$f\ssub{pr}(t) := \sum_{n=0}^m f\ssub{n}(t)$ and the {\it remainder}
as $f\ssub{rm}(t) := \sum_{n=m+1}^\infty f\ssub{n}(t)$ so that
$f(t) = f\ssub{pr}(t) + f\ssub{rm}(t)$. We obtain estimates for the
two components under the assumption that $x\ssub{0} < 0$ which means
that all $x\ssub{n}$ are negative.

First to be considered is the coefficient
$A\ssub{n} := 2|H(s\ssub{n})|/|1 + b + s\ssub{n}|$ in \eqref{term}. We
introduce the symbolism $\norm{\phi} := \sup_{t\in [-1,0]} |\phi(t)|$
to obtain from \eqref{Laplaced}
\begin{align}
  |H(s\ssub{n})| &\le |f\ssub{0}| + c\exp\{-x\ssub{n}\}\norm{\phi}
                  \int_{-1}^0 \exp\{-x\ssub{n}t\}\dv{t}\notag\\
                &\qquad\qquad < |f\ssub{0}| + c 
                  |x\ssub{n}|^{-1}\norm{\phi}\exp\{|x\ssub{n}|\}.
                  \label{H_less}
\end{align}
Also,
$|1 + b + s\ssub{n}|^2 = y^2 + (x\ssub{n}+b)^2 + 2(x\ssub{n}+b) + 1 >
B^2\exp\{-2(x\ssub{n}+b)\} + 2(x\ssub{n}+b)$ as can be seen from
\eqref{def_k}. We therefore have
\begin{equation}\label{denom_1}
  |1+b+s\ssub{n}| > c \exp\{|x\ssub{n}|\} \text{ if }x\ssub{n}+b \ge 0.
\end{equation}
Cases where $x\ssub{n}+b < 0$ are treated differently. We
(temporarily) set $u\ssub{n} = - (x\ssub{n}+b) > 0$ to obtain $|1 + b
+ s\ssub{n}|^2 > B^2\exp\{u\ssub{n}\}[1 -
2B^{-2}u\ssub{n}\exp\{-2u\ssub{n}\}]$. From the inequality
$2u\exp\{-2u\} \le e^{-1}$ we obtain
\begin{equation}\label{denom_2}
  |1+b+s\ssub{n}| > cE \exp\{|x\ssub{n}|\} \text{ if }x\ssub{n}+b < 0.
\end{equation}
with $E^2 = 1 - B^{-2}e^{-1}$. Combination of \eqref{H_less} with
\eqref{denom_1} and \eqref{denom_2} yields
\begin{align}
  & \tfrac{1}{2}A\ssub{n} < (c|x\ssub{0}|)^{-1}\exp\{-|x\ssub{n}|\}|f\ssub{0}| +
    \norm{\phi} \text{ if } x\ssub{n}+b \ge 0;\label{case_i}\\
  &\tfrac{1}{2}A\ssub{n} < E^{-1}[c^{-1}|f\ssub{0}|\exp\{-|x\ssub{n}|\} + b^{-1}
    \norm{\phi}] \text{ if } x\ssub{n}+b < 0.\label{case_ii}
\end{align}
Here use have been made of the inequalities $|x\ssub{n}| \ge
|x\ssub{0}|$ and $|x\ssub{n}| > b$ in the two different cases.

From \eqref{case_i} and \eqref{term} it is seen that
$|f\ssub{pr}(t)| < C\ssub{pr}\exp\{-|x\ssub{0}|t\}$ with $C\ssub{pr}$
a positive constant. The infinite series $f\ssub{rm}(t)$ needs more
attention. For this it is necessary to obtain information about the
behaviour of $x\ssub{n}$ for $n \ge m+1$.  Our arguments will hinge on
the equations \eqref{i} and \eqref{Curve^2} expressed in the form
$\cos y\ssub{n} = -B^{-1}(x\ssub{n}+b)\exp\{x\ssub{n}+b\}$ and
$y\ssub{n}^2 = B^2\exp\{-2(x\ssub{n}+b)\} - (x\ssub{n}+b)^2$. Since
$u\ssub{n} := -(x\ssub{n}+b) > 0$ the equations are
\begin{align}
  &\cos y\ssub{n} = B^{-1}u\ssub{n}\exp\{-u\ssub{n}\};\label{I}\\
  &y\ssub{n}^2 = B^2 \exp\{2u\ssub{n}\} - u\ssub{n}^2.\label{II}
\end{align}
The constant $\epsilon := [\arccos(B^{-1}e^{-1})]/\pi < 1/2$ will
provide some clarity. In \S\ref{Linear} it is suggested that
$y\ssub{n}$ should be near $y\lsup*\ssub{n} := (n+1/2)\pi$, and this
we shall make more precise.

From \eqref{I} we see that
$0 < \cos y\ssub{n} < B^{-1}e^{-1} = \cos(\epsilon\pi)$. Standard
trigonometry (even a good sketch) shows that these inequalities can
only be satisfied by $y\ssub{n}$ in the intervals
$(y\lsup*\ssub{n},y\lsup*\ssub{n} + \epsilon\pi]$ for $n$ odd and
$[y\lsup*\ssub{n} - \epsilon\pi,y\lsup*\ssub{n})$ for $n$ even.
The endpoints $y\ssub{n}\lsup*$ are excluded since
$\cos y\ssub{n}\lsup* = 0$.

\begin{thm}\label{estimates_x}
  If $x\ssub{n} + b <0$ then
  \begin{equation*}
    \ln[\exp\{-(b+1)\}/(n+1)\pi] < x\ssub{n} < \ln[c/n\pi].
  \end{equation*}
\end{thm}
\begin{proof}
  From \eqref{II} we see that
  $y\ssub{n}^2\exp\{-2u\ssub{n}\} =
  B^2[1-(B^{-1}u\ssub{n}\exp\{-u\ssub{n}\})^2] = B^2[1 -
  \cos^2y\ssub{n}] = B^2\sin^2y\ssub{n}$ (having used \eqref{I}
  again). This, in turn gives (after some manipulation) 
  $x\ssub{n} = \ln[c|\sin y\ssub{n}|/y\ssub{n}]$. Careful
  consideration of the cases $n$ even/odd leads to $B^{-1}e^{-1} =
  \cos\epsilon\pi \le |\sin y\ssub{n}| < 1$ and it follows that
  $
    -\ln[\exp\{b+1\}y\ssub{n}\} \le x\ssub{n} \le \ln[c/y\ssub{n}].
  $
  Since $0 < \epsilon < 1/2$ we have
  $y\lsup*\ssub{n} - \epsilon\pi > n\pi$ and
  $y\lsup*\ssub{n} + \epsilon\pi < (n+1)\pi$. It follows that
  $n\pi < y\ssub{n} < (n+1)\pi$ regardless of the parity of $n$.
\end{proof}

The inequality \eqref{case_i} may now be employed to
estimate $f\ssub{n}(t)$ in \eqref{term} when $n \le m$. This results in
\begin{equation}\label{est_3}
  |f\ssub{n}(t)| <
    2[(c|x\ssub{0})^{-1}|f\ssub{0}|\exp\{-|x\ssub{n}|\}
    + \norm{\phi}] \exp\{-|x\ssub{n}|t\} \text{ if }n\le m.
\end{equation}
For the case $n > m$ we apply \eqref{case_ii} together with
Thm. \ref{estimates_x} to obtain
\begin{align}
  |f\ssub{n}(t)| &<
    2E^{-1}\left[c^{-1}|f\ssub{0}|\left(\frac{c}{\pi}\right)^{t+1}
                   \left(\frac{1}{n}\right)^{t+1}\right.\notag\\
                &\hspace{70pt} + b^{-1}\norm{\phi}\left.
    \left(\frac{c}{\pi}\right)^t\left(\frac{1}{n}\right)^t\right]
      \text{ if }n > m.\label{est_4}
\end{align}

The solution $f$ of the linearized homogeneous problem can now be
estimated in terms of the parameters of the problem:
\begin{thm}\label{est_f}
  If $x\ssub{0} < 0$ then:
  \begin{enumerate}
    \item[\textnormal{(a)}] There exists a constant $C\ssub{pr}>0$
      such that
      \begin{equation*}
        |f\ssub{pr}(t)| < C\ssub{pr}\exp\{-|x\ssub{0}|t\}.
      \end{equation*}
    \item[\textnormal{(b)}] There exist positive constants
      $C\ssub{rm,1}, C\ssub{rm,2}$ such that for $t>1$
      \begin{equation*}
        |f\ssub{rm}(t)| < C\ssub{rm,1}\zeta(t+1)
        \left(\frac{c}{\pi}\right)^{t+1} + C\ssub{rm,2}\zeta(t)
        \left(\frac{c}{\pi}\right)^t  
      \end{equation*}
      with $\zeta$ the Euler-Riemann zeta function.
  \end{enumerate}
\end{thm}
\begin{cor}\label{decay_1}
  If $x\ssub{0} < 0$ and $c < \pi$, $f(t) \to 0$ uniformly and
  exponentially as $t \to \infty$.
\end{cor}
We note that the convergence is determined by $c/\pi$ which factors
out of the partial sums in the series expansion. 

\section{Non-linear behaviour}
\label{nonlinear} 

We now turn to the non-linear equation \eqref{split}, with the usual
initial condition, expressed in the following way:
\begin{equation*}%\label{nonlin-prob}
  \left.
    \begin{aligned}
      &f'(t) + [bf(t) + cf(t-1)] = F(f(t),f(t-1)),\quad t > 0;\\
      &F(u,v) := -u[bu + cv].
    \end{aligned}
  \right\}
\end{equation*}
Our first aim is to obtain a non-linear integral
representation. Towards this we consider the non-homogeneous problem
\eqref{non-hom} once again by writing\break $f = f\ssub{[0]} + f\ssub{[v]}$
with $f\ssub{[0]}$ the solution of the homogeneous problem studied in
\S\ref{Linear} and $f\ssub{[v]}$ the solution of the non-homogeneous
problem under the homogeneous initial condition $f\ssub{[v]}(t) = 0$
for $-1\le t < 0$. For this purpose we introduce the kernel function
$\mathfrak{K}$ as the solution of the problem
\begin{equation}\label{def_kernel}
  \left.
    \begin{aligned}
      &\mathfrak{K}'(t) + b\mathfrak{K}(t) +
          c\mathfrak{K}(t-1)=0,\text{ for } t>0;\\
      &\mathfrak{K}(t) = 0, \text{ for } t < 0;\\
      &\mathfrak{K}(0) = 1.    
    \end{aligned}
  \right\}
\end{equation}
This differs from the homogeneous problems studied in \S\S
\ref{Linear}--\ref{estimates} in the jump discontinuity at $t=0$ which
is not serious. In fact, $\mathfrak{K}(t) = \exp\{-bt\}$ for
$0 \le t \le 1$. Theorems \ref{estimates_x} and \ref{est_f} apply in
this case as well with $f\ssub{0} = 1$ and $\norm{\phi} = 0$. If
$\lambda := \ln\{\pi/c\}$ this leads to:
\begin{align*}
  |\mathfrak{K}(t)| &< 2(c|x\ssub{0}|)^{-1} \exp\{-|x\ssub{0}|(t+1)\}
                      \notag\\
                    &\hspace{50pt}+ 2(Ec)^{-1}\zeta(t+1)\exp\{-\lambda(t+1)\}
                      \text{ for }t>1.
\end{align*}
Therefore there are constants $C>0$, $\mu > 0$ such that
\begin{equation}
  |\mathfrak{K}(t)| \le C\exp\{-\mu t\} \text{ for }t \ge 0.
     \label{est_k}
\end{equation}

Repetition of the Laplace transform procedure of \S\ref{Linear} for
$f\ssub{[v]}$ yields, in the notation of \eqref{Laplaced},
$
  \widehat{f\ssub{[v]}}(s) = \hat{v}(s)/h(s).
$

From the convolution theorem we now have
$
f\ssub{[v]}(t) = \int_0^t\mathfrak{K}(t-t')v(t')\dv{t'}
 = \int_0^t\mathfrak{K}(t')v(t-t')\dv{t'},
$
and hence,
\begin{equation}\label{f0_plus-fv}
  f(t) = f\ssub{[0]}(t) + \int_0^t\mathfrak{K}(t')v(t-t')\dv{t'}.
\end{equation}

The formal calculations above can be justified by
Thm. \ref{est_non-hom} and the (tacit) assumption that $v$ has a
Laplace transform. We continue with the formalism by letting $v(t) =
F(f(t),f(t-1))$. 

To see that the integral representation \eqref{f0_plus-fv} is more
than formal, we note that by virtue of Thm.\ref{Upper-bd}(a) (with
$f$, $f\ssub{0}$ replaced by $f-1$, $f\ssub{0}-1$), that the function
$t \to v(t)$ is bounded. In fact, if $|f(t)| \le M$ then
\begin{equation}\label{bound_F}
  |v(t)| \le \tilde{\beta}M|f(t)| \le \tilde{\beta}M^2,
\end{equation}
so that the Laplace transform exists.

It is now possible to consider the asymptotic stability of the
nonlinear equation \eqref{split} under the initial condition $f(t) =
\phi(t)$ for $-1 \le t \le 0$.

\begin{thm}\label{stability}
  If\  $\lim_{t\to\infty}v(t)$ exists, $x\ssub{0} < 0$ and $c < \pi$,
  then $f(t) - f\ssub{[0]} \to 0$ as $t \to \infty$. 
\end{thm}
\begin{proof}
  From \eqref{f0_plus-fv} we see that
  \begin{align}
    f\ssub{[v]}(r+t) - f\ssub{[v]}(r)
      &= \int_r^{r+t} \mathfrak{K}(t')v(r+t-t')\dv{t'}\notag\\
        &\qquad+ \int_0^r \mathfrak{K}(t')[v(r+t-t')-v(r-t')]\dv{t'}.
           \label{Diff}
  \end{align}
  To estimate the terms on the right we rely on the inequalities
  \eqref{est_k} and \eqref{bound_F}.

  For the first term on the right of \eqref{Diff} one has
  \begin{equation*}
    \left|\int_r^{r+t}\mathfrak{K}(t')v(r+t-t')\dv{t'}\right|
      \le \tilde{\beta}M^2C\mu^{-1}\exp\{-\mu r\}[1 - \exp\{-\mu t\}]
  \end{equation*}
  and this tends to zero as $r \to \infty$.

  The second term is treated differently by splitting the integral in
  two, one over $[0,r/2]$ and the other over $[r/2,r]$. Since the
  limit of $v(t)$ exists, there is for given $\varepsilon > 0$,
  $t\ssub{\varepsilon}$ such that for $r,t > t\ssub{\varepsilon}$,
  $|v(r+t) - v(r)| < \varepsilon$. Now
  \begin{align*}
    &\left|\int_0^{r/2}
      \mathfrak{K}(t')[v(r+t-t')-v(r-t')]\dv{t'}\right|\\
        &\hspace{80pt}\le  C \mu^{-1}[1-\exp\{-\mu r/2\}]\varepsilon
          \text{ for } r,t > 2 t\ssub{\varepsilon}.
  \end{align*}
  The integral over $[r/2,r]$ is treated the same as the first term on
  the right of \eqref{Diff} and we conclude that
  $\lim_{t\to\infty}[f(t) - f\ssub{[0]}(t)]$ exists. The conclusion
  follows from Cor. \ref{decay_1} and Thm. \ref{Upper-bd}.
\end{proof}
The heuristic argument in \S\ref{linearize} for the linear homogeneous
equation with solution $f\ssub{[0]}$ is justified under the hypotheses
of Thm. \ref{stability}. In fact, since $f\ssub{[0]} \to 0$ as
$t \to \infty$, the theorem shows that under a fairly weak hypothesis
the nonlinear part in the representation \eqref{f0_plus-fv} decays
to zero for large $t$. More explicit conditions for asymptotic
stability are known. In \cite[Chap.11]{Bellman:1963} Lyapunov
stability of general delay equations is demonstrated under the
assumptions that $x\ssub{0} < 0$ and the initial state $\phi$ is near
the equilibrium level. In \cite[Thm. 3]{Ruan:2006} it is shown, by
construction of a Lyapunov function, that asymptotic equilibrium occurs
if $c < b$. 

\section{An example}
\label{example}

After the somewhat daunting mathematical sections above, it is
appropriate to give an example. The one presented here is taken from
an epidemic which at the time of this writing was very much on every
mind. Local data suggests that the theta-model (\S\S \ref{verhulst},
\ref{Norm}) is necessary. Least squares estimates based on early data
suggest the parameter values $\theta = 2.8$, $\beta = 0.017$/day and
$f\ssub{0} = 0.046$. The parameters $\tau$ and $\gamma$ have been
manipulated experimentally to obtain results that correspond
reasonably to perceptions. The choices are $\tau = 32$ days,
$\gamma = 0.9$/day. This gives $B = c\Exp{b} = 1.53778 < \pi/2$ and
$c/\pi = 0.475\dots$ All zeros have negative real part. In fact, a
numerical computation based on \eqref{Alt_eqn} yields
$x\ssub{0} = -0.043\dots < -b = -0.028\dots$ and
$y\ssub{0} = 1.561\dots$

Numerical solution of the initial value problem \eqref{expr1},
\eqref{init-cond}, with the initial state taken as constant, namely
$\phi \equiv f\ssub{0}$, (remembering that $f$ really means
$f^\theta$) resulted in Fig.\ref{Example}. The period of undulation,
estimated from the linearization, according to \eqref{period}, is
$T \approx 2\pi\tau/y\ssub{0} = 128.9$ days. The equilibrium level
corresponds to $F\ssub{\infty} = [\beta\ssup{*}]^{1/\theta} = 0.238$.
\begin{figure}[h!]
 \centering
 \includegraphics[width=0.75\textwidth]{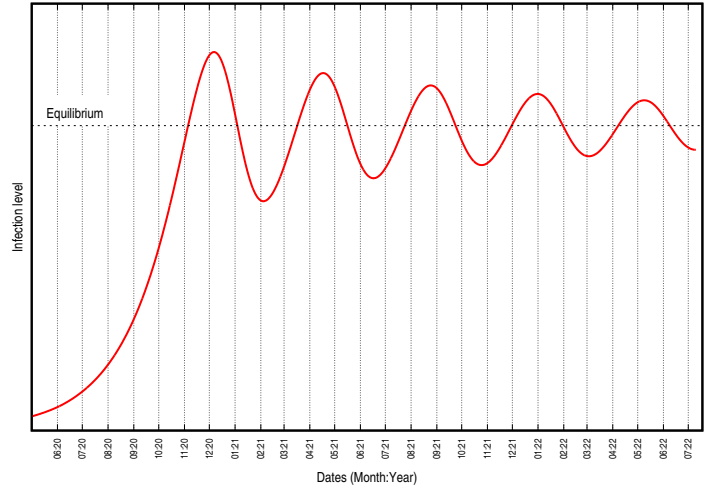} 
 \caption{{\footnotesize Undulation and decay}}
 \label{Example}
\end{figure}

The computation reported above shows undulation and suggests decay of
infectivity levels as time goes on. Our analysis of the linear problem
in \S\ref{Linear} indicates that decay could be to the asymptotic
equilibrium state $f\ssub{\infty} = 1$. We have also shown that by
increasing only $\tau$ a zero with non-negative real part can occur
and then, according to the linearized version, there will be no
decay. According to \eqref{crit-tau} the critical value in the present
example is $\tau = \tau\ssub{0} = 34.033\dots$ days. For this value of
$\tau$, $x\ssub{0} = 0$. The question is: would this be so for the
non-linear equation? 

First, Thm. \ref{Upper-bd}(a) states that infectivity levels cannot run
away, but it can happen in the linearization (according to
\eqref{term}). To come closer to answers it is instructive to
calculate trajectories in the phase portrait ($f'$ as a function of
$f$). This is shown in Fig.\ref{Example+} for the case discussed above
and with $\tau$ alone increased to 35, slightly above
$\tau\ssub{0}$. The result seems to confirm that in one case
($\tau=32$) decay is to the asymptotic point $(1,0)$ and in the other
case ($\tau = 35$) to a limit cycle about this point. The
linearization indeed leads to some clarification.

\begin{figure}[h!]
 \hspace{50pt}
 \includegraphics[width=0.68\textwidth]{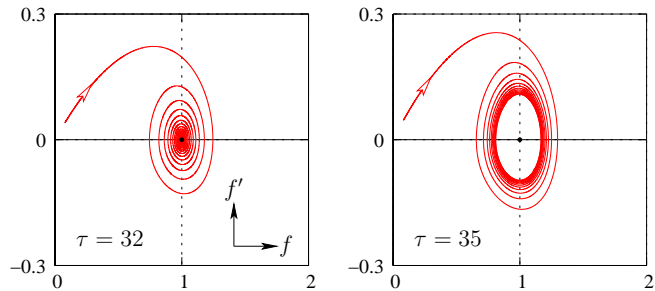} 
 \caption{{\footnotesize Decay to equilibrium (left) and a
     limit cycle (right)}}
 \label{Example+}
\end{figure}

\section{The numerical algorithm}
\label{Numeric}

A brief note on the method used for the numerical calculations is in
order. We refer to the representation obtained in \S\ref{general} by
the substitution $f = 1/g$ and specifically note the expressions
\eqref{intf1} and \eqref{tranfd_eqn} used to obtain qualitative
results. For the numerical solution we integrate the latter expression
over the interval $[t,t+\eta]$ to obtain:
\begin{align*}
  g(t+\eta) &= \exp\{-\delta i(t)\}g(t)
  + \beta\ssup{*}[1-\exp\{-\delta i(t)\}]\\
  &\qquad + \tilde{\beta}\beta\ssup{*}\gamma\lsup{*}
  \exp\{-i(t+\eta)\}\int_t^{t+\eta} \exp\{i(t')\}f(t'-1)\dv{t'},
\end{align*}
with
\begin{equation*}
  \delta i(t) := i(t+\eta) - i(t) = \tilde{\beta}[\eta 
        - \gamma\lsup{*}\int_t^{t+\eta}f(t'-1)\dv{t'}].
\end{equation*}
Approximation of the integrals above by the trapezium rule gives
\begin{align*}
  &\int_t^{t+\eta} \exp\{i(t')\}f(t'-1)\dv{t'}\\
  &\qquad\qquad\approx \tfrac{1}{2}\eta [\exp\{i(t+\eta)\}
       f(t+\eta -1) + \exp\{i(t)\}f(t-1)];\\
    &\hspace{22pt}\delta i(t) \approx \tilde{\beta}\eta\{1 - \tfrac{1}{2}
    \gamma\lsup*[f(t+\eta -1)+f(t-1)]\}.
\end{align*}
Combination of everything results in the computational algorithm 
\begin{align*}
  g(t+\eta) &\approx g(t)\exp\{-\delta i(t)\}
  + \beta\ssup{*}[1-\exp\{-\delta i(t)\}]\\
  &\qquad + \tfrac{1}{2}\eta\tilde{\beta}\beta\ssup{*}\gamma\lsup{*}
  [f(t+\eta - 1) + \exp\{-\delta i(t)\}f(t-1)].
\end{align*} 
Since $f(t)$ is given for $-1 \le t \le 0$, this can be computed
(coding is straightforward). The algorithm is grounded in the problem.

% %%

\section{Concluding unscientific  remarks}
\label{Unscientific}

The caption above is borrowed from a similar-sounding title by
Johannes Climacus (S\o ren Kierkegaard), published in 1846 which, in a
way, echoes the Socratic aphorism that {\it the only wisdom we have is
  knowing that we do not know} --- a paradox that can be (partly)
resolved if `knowing' is replaced by `understanding'. Understanding,
it has been said, expands when horizons of knowing meet. It is
never complete.

A fundamental tenet for the mathematical description of growth is the
Verhulst logistic model which states that growth is determined by what
is left to grow upon. It has the property that growth will increase to
devour available resources. If the initial level is below equilibrium
levels will increase towards equilibrium. On the other hand if the
level is initially above equilibrium it will decrease towards
equilibrium. Since $\beta\ssup* = 1 - \gamma\lsup*$, the
logistic-recovery equation \eqref{expr1} may be re-phrased as
\begin{equation*}
  f'(t) = \tilde{\beta} f(t) [1 - f(t)] + c f(t)[f(t)-f(t-1)],
\end{equation*}
which is the logistic equation perturbed by a recovery term that
occurs in {\it Van der Plank's equation} \eqref{vdPlank} without
dormancy. Recovery has the effect of overshooting the equilibrium
state. Logistic growth, acting like a counterweight, then forces
growth to decrease. Once the level is below equilibrium, growth will
turn upwards again. This cyclic process can decay towards stable
equilibrium, but may also become repetitive like the motion of a
pendulum or a planet orbiting the sun. 

It is interesting to note that unfettered undulation can only occur
when $\rho = 1/R\ssub{0} > 1$ (Thm.\ref{About_roots}). The idea that
$R\ssub{0} < 1$ is prudent, may be questioned. Within the present
discussion this can lead to lowering of the equilibrium level, but at
the price of an undulation in which equilibrium could be a
spectre. But then, $R\ssub{0}$ in the classical SIR model cannot have
the same meaning as $1/\rho$ in the model discussed here.

The notion of recovery as used here should not be confused with the
clinical use of the word. One might ask: what is recovering, the
patient or the pathogen? The long recovery period (32 days) used in
the example makes the question more incisive; so does the loss of
asymptotic equilibrium when recovery takes longer. One could argue
that a sufficiently short recovery period provides less opportunity
for transmission (or evolution) of the pathogen so that decay to
equilibrium would be the result.  We should take heed of the view of
Dr. James van der Plank in \cite{vdplank:1965}: observations of
infected subjects merely reflect the state of the pathogen.

Mathematical models serve as a basis for motivated speculation and not
much more. They count among the many metaphors we invent to explain
and understand what is called {\it reality}. Computational experiments
with such models, not supported by mathematical insight, are similar
to searching for ``\dots two grains of wheat hid in two bushels of
chaff \dots''. The search for grains of truth may be long and arduous.

\bibliographystyle{plainurl}      % mathematics and physical sciences

%\bibliography{undulation}   % name your BibTeX data basearxivarxivarchivaaaaaaa

\begin{thebibliography}{10}

\bibitem{Bellman:1943}
R.E. Bellman.
\newblock The {S}tability of {S}olutions of {L}inear {D}ifferential
  {E}quations.
\newblock {\em Duke Math. J.}, 10:643--647, 1943.

\bibitem{Bellman:1963}
R.E. Bellman and K.E. Cooke.
\newblock {\em Differential-{D}ifference {E}quations}.
\newblock Academic Press, 1963.

\bibitem{Buonomo:2015}
B.~Buonomo and M.~Cerasuolo.
\newblock The {E}ffect of {T}ime {D}elay in {P}lant-{P}athogen {I}nteractions
  with {H}ost {D}emography.
\newblock {\em Math Biosc. Eng.}, 12:473--490, 2015.

\bibitem{Cacciapaglia:2022}
G.~Cacciapaglia, C.~Cot, and F.~Sannino.
\newblock Multiwave pandemic dynamics explained: how to tame the next wave of
  infectious diseases, March 2021.
\newblock In Nature: Scientific Reports, Vol. 11.
  \url{https://www.nature.com/articles/s41598-021-85875-2}.

\bibitem{Cunningham:1954}
W.J. Cunningham.
\newblock A nonlinear difference-differential equation of growth.
\newblock {\em Proc. Nat. Acad. Sci.}, 40:709--713, 1954.

\bibitem{DellaMorte:2021}
M.~Della~Morte and F.~Sannino.
\newblock Renormalization group approach to pandemics as a time-dependent {SIR}
  model, January 2021.
\newblock In Frontiers in Physics, Vol. 8.
  \url{https://doi.org/10.3389/fphy.2020.591876}.

\bibitem{Feller:1966}
W.~Feller.
\newblock {\em An {I}ntroduction to {P}robability {T}heory and {I}ts
  {A}pplications}, volume~II.
\newblock John Wiley \& Sons, Inc., 1966.

\bibitem{Gilpin:1973}
M.E. Gilpin and F.J. Ayala.
\newblock Global {M}odels of {G}rowth and {C}ompetition.
\newblock {\em Proc. Nat. Acad. Sci. USA}, 70:3590--3593, 1973.

\bibitem{Gopalsamy:1992}
K.~Gopalsamy.
\newblock {\em Stability and {O}scillations in {D}elay {D}ifferential
  {E}quations}, volume~74 of {\em Mathematics and its {A}pplications}.
\newblock Springer, Dordrecht, 1992.

\bibitem{Hutchinson:1948}
G.E. Hutchinson.
\newblock Circular {C}ausal {S}ystems in {E}cology.
\newblock {\em Ann. New York Acad. Sci.}, 50:221--246, 1948.

\bibitem{Jones:1961}
G.S. Jones.
\newblock Asymptotic behavior and periodic solutions of a nonlinear
  differential-difference equation.
\newblock {\em Proc. Nat. Acad. Sci.}, 47:879--882, 1961.

\bibitem{Kermack:1927}
W.O. Kermack and A.G. McKendrick.
\newblock A {C}ontribution to the {M}athematical {T}heory of {E}pidemics.
\newblock {\em Proc. Roy. Soc. London. Ser. A}, {115}:700--721, 1927.

\bibitem{Reiser:2020}
P.A. Reiser.
\newblock Modified {SIR} {M}odel {Y}ielding a {L}ogistic {S}olution.
\newblock Preprint, 2021.
\newblock ArXiv. \url{https://doi.org/10.48550/arXiv.2006.01550}.

\bibitem{Ruan:2006}
S.~Ruan.
\newblock Delay {D}ifferential {E}quations in {S}ingle {S}pecies {D}ynamics.
\newblock In O.~Arino, M.~Hbid, and E.A. Dads, editors, {\em Delay
  {D}ifferential {E}quations and {A}pplications}, volume 205 of {\em NATO
  Science Series (II. Mathematics, Physics and Chemistry)}, pages 477--517.
  Springer, Dordrecht, 2006.

\bibitem{vdplank:1963}
J.E. Van~der {P}lank.
\newblock {\em Plant {D}iseases. {E}pidemics and {C}ontrol}.
\newblock Academic Press, New York, 1963.

\bibitem{vdplank:1965}
J.E. Van~der {P}lank.
\newblock Dynamics of {P}lant {D}isease.
\newblock {\em Science}, 147:120--124, 1965.

\bibitem{Verhulst:1838}
P-F Verhulst.
\newblock Notice sur la loi que la population suit dans son accroissement.
\newblock {\em Correspondance math\'ematique et physique}, 10:113--121, 1838.

\bibitem{Wright:1955}
E.M. Wright.
\newblock A nonlinear difference-differential equation.
\newblock {\em J. Reine Angew. Math.}, 194:66--87, 1955.

\end{thebibliography}

\end{document}